\documentclass[10pt,a4paper]{article}
\usepackage{amssymb,amsmath, amsfonts}
\usepackage{stix}
\usepackage{graphicx,graphics}
\usepackage{mathtools}
\usepackage[english]{babel}
\usepackage[utf8]{inputenc}
\usepackage{csquotes}
\usepackage{epsfig,url}
\usepackage{bbm}
\usepackage{theorem}
\usepackage{a4wide}
\usepackage{enumerate}
\usepackage{verbatim}
\usepackage{color}
\usepackage{esint}

\usepackage[T1]{fontenc}

 \usepackage[colorlinks=true, pdfstartview=FitV, linkcolor=blue, citecolor=blue, urlcolor=blue,pagebackref=false]{hyperref}

\newcommand{\ud}{\;\mathrm{d}}
\newcommand{\uS}{\;\mathrm{dS}}
\newcommand{\uud}{\mathrm{d}}

\providecommand{\St} {\mathbb{S}^2}

\providecommand{\DD}{\mathbb{D}}

\providecommand{\eps}{{r}}

\providecommand{\supp}{\operatorname{supp}}
\providecommand{\exp}{\operatorname{exp}}

\newtheorem{theorem}{Theorem}[section]
\newtheorem{corollary}[theorem]{Corollary}
\newtheorem{proposition}[theorem]{Proposition}
\newtheorem{definition}[theorem]{Definition}
\newtheorem{lemma}[theorem]{Lemma}

{\theorembodyfont{\upshape}
\newtheorem{remark}[theorem]{Remark}

}
\numberwithin{equation}{section}
\numberwithin{theorem}{section}
\newcommand{\qed}{\hfill$\Box$}
\newenvironment{proof}{\begin{trivlist}\item[]{\em Proof:}\/}{\qed\end{trivlist}}
\newenvironment{proofof}[1][Proof]{\noindent \textit{#1.} }{\ \qed}

\newcommand{\Reals}{{\mathbb R}}

\newcommand{\Naturals}{{\mathbb N}}

\newcommand{\cf}{{\mathbbm 1}}

\newcommand{\Exp}{\operatorname{Exp}}
\newcommand{\Uni}{\operatorname{Uni}}

\definecolor{darkblue}{rgb}{0,0,0.7}
\definecolor{darkgreen}{rgb}{0,0.5,0}

\title{Improved geometric and recollision estimates for the invariance principle of the random Lorentz gas }

\date{\today}

\newcommand{\email}[1]{E-mail: \tt #1}
\newcommand{\emailKarsten}{\email{k.matthies@bath.ac.uk}}
\newcommand{\emailraphael}{\email{winterr6@cardiff.ac.uk}}

\newcommand{\UBaddress}{\em University of Bath
	Department of Mathematical Sciences, UK}
\newcommand{\CUaddress}{\em Cardiff University, School of Mathematics, UK}

\author{
Karsten Matthies  \thanks{{\emailKarsten.} }, Raphael Winter \thanks{\emailraphael} \\[1em]
$\,^*$\UBaddress \\[0.5em]
$\,^\dag$\CUaddress}

\date{\today}

\begin{document}

\maketitle
\begin{abstract}
  By improving geometric recollision estimates for a random Lorentz gas, we extend the timescale $T(r)$ of the invariance principle for a Lorentz gas with particle size $r$ obtained by Lutsko and Toth (2020) from $\lim_{r \rightarrow 0} T(r) r^{2} |\log(r)|^2 =0$  to $\lim_{r \rightarrow 0} T(r) r^{2} |\log(r)| =0$. We show that this is the maximal reachable timescale with the coupling of stochastic processes introduced in the original result. In our improved geometric estimates we make use of the convexity of scatterers to obtain better dispersive estimates for the associated billiard map. We provide additional estimates which potentially open the possibility to reach, with a more elaborate coupling argument, up to timescales just below of $T(r)\sim r^{-2}$ when recollision patterns of arbitrary length occur.
\end{abstract}


\section{Introduction}

The dynamical behaviour of a hard-sphere flow between fixed scatterers is a classical topic in the theory of billiards as well as in mathematical statistical physics for the emergence of diffusion. In this paper, we use billiard maps to extend the time scales for which diffusive behaviour can be established.

We consider the random Lorentz gas in $\Reals^3$, i.e. balls of radius $r$, infinite mass and with centres distributed by a Poisson point process of intensity $\rho$ such that $r^3 \rho$ is sufficiently small. We are interested in the trajectory $t \mapsto X^{r,\rho} \in \Reals^3$ of a point particle starting at the origin in a random unit direction, with free flights between collisions and elastic collisions with the scatterers.   For fixed times $t \in [0,T]$, the process  $t \mapsto X^{r,\rho}$ converges in a Boltzmann-Grad scaling ($r\to 0$, $\rho \to \infty$ and $r^2 \rho \to 1$)
to a Markovian flight process and the distribution can be described by kinetic equations as shown originally by Gallavotti \cite{Gal70}, Spohn \cite{Spohn78} and Boldirghini-Munimovich-Sinai \cite{BBS83}. For more recent results, see Golse~\cite{golse22} and references therein. The main difference between the random Lorentz gas and a random flight process is that the latter is independent of its past, so it ignores the positions of previous scatterers and it allows collisions at places where in the past had not been any scatterers (\emph{shadowed scattering}). Additionally, the point particle in a Lorentz gas can have \emph{recollisions} with previous scatterers. These are the two main effects that need to be controlled in convergence proofs.   

Diffusive behaviour can be observed for the free-flight process after a further rescaling 
\begin{align}
    \label{eqn:diffscal} t \mapsto \frac{X(Tt)}{\sqrt{T}} \quad \mbox{ for }    T \to \infty,
\end{align}
such that even better control of shadowed scattering and recollisions is necessary when this is to be extended to the Lorentz gas. Lutsko and T\'oth \cite{LT20} prove an invariance principle (i.e. diffusive behaviour in the \eqref{eqn:diffscal} scaling) in a combined Boltzmann-Grad and diffusive limit $r\to 0$, $\rho \to \infty$, $r^2 \rho \to 1$ and $T(r) \to \infty$, such that 
$\lim_{r \rightarrow 0} T(r) r^{2} |\log(r)|^2 =0$.
The convergence holds when averaging over the random initial velocity and random placement of scatterers.
The key technical innovation is to use probabilistic coupling arguments involving three processes: the Lorentz gas $X$, the Markovian flight process $Y$ and an auxiliary process $Z$ that has some memory, but only until the previous collision event. Lutsko and T\'oth combine the probabilistic arguments with an analysis of the three dimensional geometry of the tracer particle bouncing between two scatterers of the same size. 

In this paper, we give improved estimates by analysing a billiard map on the two scatterers and we extend the time-scale of the  combined Boltzmann-Grad and diffusive limit to
$T(r) \to \infty$ under the condition  $\lim_{r \rightarrow 0} T(r) r^{2} |\log(r)| =0$. 
We note that Lutsko and T\'oth achieved such an improved time-scale in the simpler geometry of the wind-tree process \cite{Toth21}. 
In our analysis of the Lorentz gas, a separation of the two cases, where the distance between the scatterers is  smaller or larger than a fixed multiple of their radius, is advantageous. Several effects  that are estimated to be of
order $r^{2} |\log(r)|^2$ in \cite{LT20} will be shown to be of order $r^2$. We do not re-analyse other effects that lead to errors of $r^{2} |\log(r)|$. 
This immediately allows us to use the same methods as in \cite{LT20}  as well as gaining further insights about the dynamic behaviour. In particular, we also show that after a simple recollision the exit velocity distribution regains independence of the initial velocity  and is $\Uni(\St)$ distributed in the critical case when the distance between the scatterers is large, see Lemma~\ref{lem:independence}.

We stress the fact that the result presented here realizes the maximal timescale attainable with bounds on direct shadowing/recollision events. More precisely, the coupling of the $Z$-process and the Lorentz gas process $X$ fails beyond the timescales determined by  $\lim_{r \rightarrow 0} T(r) r^{2} |\log(r)|=0$. This follows from the observation that mechanical inconsistencies after two collisions (i.e. indirect) appear beyond this timescale and are not respected by the $Z$-process, see Remark~\ref{rem:indirect}. The same holds for inconsistencies between two subsequent flights in the presence of a direct shadowing event. 
An interesting question for future research is whether a more elaborate auxiliary process, possibly one having memory of the last two scatterers, would allow us to extend the results to times $T(t)$ under the condition $\lim_{r \rightarrow 0} T(r) r^{2} =0$.  Finally, we recall that on timescales 
$T(r)\sim r^{-2}$ complex recollision patterns of arbitrary length seem to be present that likely require entirely new tools, see also 
\cite[Remarks on dimension (3)(b)]{LT20}.

\paragraph{Related work:}
The statistical and long-term diffusive behaviour has been analysed for various models related to the random Lorentz gas. T\'oth recently showed the invariance principle for almost all realisations just suitably averaging over initial velocities \cite{toth25}. Lutsko and T\'oth adapted their ideas to the wind-tree process \cite{Toth21} and to Lorentz gases with an underlying magnetic field \cite{lutsko24}. Nota, Nowak and Saffirio also analyse the magnetic Lorentz gas in \cite{nota24}. Different scaling behaviour for the Lorentz gas can be observed for periodically placed scatterers, see the works by Caglioti-Golse \cite{Cag10} and Marklof-Str\"ombergson \cite{marklof08,Marklof11,Marklof-24}. In this case, a kinetic description by linear Boltzmann equations is not possible, and for longer times a superdiffusive behaviour can be observed \cite{Marklof16} as well as a non-standard central limit theorem \cite{Balint23}. However, Wennberg was able to derive a linear Boltzmann equation for a Lorentz Gas with nearly periodic scatterer distribution, see \cite{wennberg23}. Diffusive behaviour has also been shown for a particle in a suitable random forcefield on $\Reals^d$ for $d\geq 3$ in \cite{Durr-87}  and on $\Reals^2$ in \cite{Kesten81}. 

In an idealised Rayleigh gas the background particles are of finite mass and move, but do not interact with each other. In the kinetic limit, a linear Boltzmann equation can be derived \cite{spohn91} for hard-sphere interactions, for a general conservative situation see \cite{Matthies18} and with annihilation \cite{Nota19}. In the conservative setting, if the background particles are Maxwell-distributed, diffusive behaviour can be observed on time-scales below $\eps^{-1/8}$, see \cite{MatthiesSyntaka2024surrey}. Fractional diffusion can be derived for suitable short-range interactions between the tagged particle and background particles  with a fat-tailed distributions for time-scales below $\eps^{-\alpha}$ for some small $\alpha>0$, see \cite{matthies24}.

In a non-idealised Rayleigh gas, the background particles  interact with each other, but are assumed to be in equilibrium. In this considerably more complex system, diffusive behaviour has been shown in \cite{bod15brown} for a double logarithmic timescale, which has recently been improved to $(\log|c_{\beta} \log \eps| )^{1/2-\alpha}$  in \cite{Fou24}.   More general fluctuations around the equilibrium in the full hard sphere 
flow lead to fluctuating Boltzmann equations, see the recent woks by Bodineau, Gallagher, Saint-Raymond, and Simonella
\cite{Bodineau23,Bodineau23-2,Bodineau24}. For more details on scaling limits of tagged particles in random fields, see~\cite{NSW21,NSW22}.

\paragraph{Plan of paper:} We introduce the setup for direct recollisions in the next section. The main improved geometric estimate is given in Subsection~\ref{ssec:impro}, this will yield a better control of the probability of scattering events in Subsection~\ref{ssec:prob}. Additional independence results are shown in Subsection \ref{ssec:indep}.  The results are applied to the invariance principle in Section~\ref{sec:LT} by sketching the relevant improvements of~\cite{LT20}.

\section{Recollisions of the tracer particle with two scatterers}

In this section, we present the geometric improvements on the probability of collision events. In order to keep the presentation self-contained, we review the main elements of the geometric argument introduced in~\cite{LT20}. We use variables $(u,\xi,v) \in  \DD=\St\times \Reals^+ \times \St $
to parametrise possible direct recollisions, i.e.  backscattering of the tracer particle to an obstacle after the subsequent collision. Without loss of generality we fix the incoming velocity for the first collision to be the first unit vector $e=e_1\in \Reals^3$. In the set of variables $(u,\xi,v) \in  \DD=\St\times \Reals^+ \times \St $, the velocity $u$ then represents the velocity after the first collision, $v$  the velocity after the second collision, and $\xi$ the free flight time in between.  

\begin{figure}
    \centering
    \includegraphics[width=\linewidth]{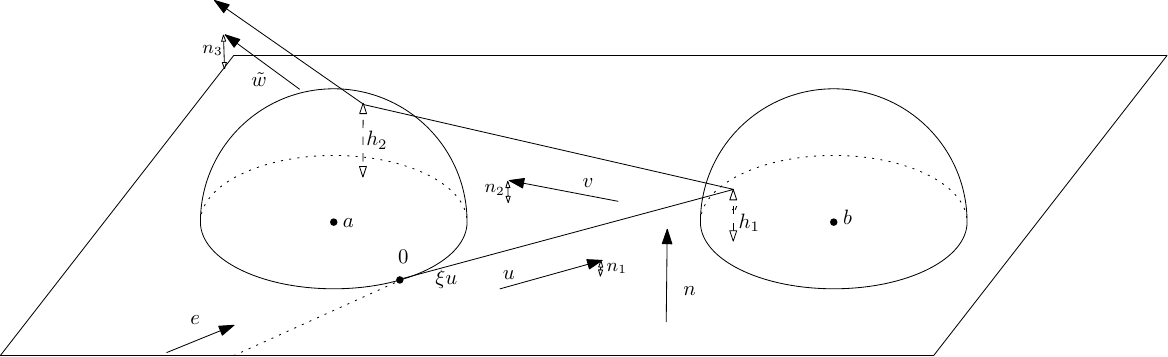}
    \caption{Notation for the motion between two scatterers.}
    \label{fig:setup}
\end{figure}

\begin{definition}[Recollision events] \label{def:recollisions}
Given a radius $r>0$, a recollision event is determined by $(u,\xi,v)\in \DD $, where $\DD=\St\times \Reals^+ \times \St$.  For $(u,\xi,v)\in \DD $,  we define the process $\tilde{Z}_r(t)\in \Reals^3$, $t\in \Reals$ by
\begin{enumerate}
	\item $\tilde{Z}_r(0) = 0$, and $\tilde{Z}_r(t) = te$ for all $t<0$, where $e=e_1$ is the first coordinate vector.
	\item $\tilde{Z}_r(t)= tu$, for $t\in (0,\xi)$.
	\item The centres $a,b\in \Reals^3$ of the two obstacles are defined by
	\begin{align} \label{eq:centers}
		a:= r \frac{e-u}{|e-u|}, \quad b:= \xi u + r\frac{u-v}{|u-v|}.
	\end{align}
	\item $\tilde{Z}_r(t)$ for $t\geq \xi$ given by dynamics of a point particle interacting with the infinite hard-sphere obstacles $B_r(a)$ and $B_r(b)$ with $\dot{\tilde{Z}}_r(\xi^+)=v$.
	\item Let $N$ be the total number of collisions of this process, and denote by $w_k$, $0 \leq k\leq N$ the velocity after the $k$-th collision with $w_0=e$, $w_1=u$, $w_2=v$. Similarly, define $0=\tau_1 \leq \tau_1\leq \cdots \leq \tau_N$ as the collision times of the process.
	\item We denote by $\tilde{\beta}$ the time of the last collision and $\tilde{w}=w_N$ the outgoing velocity.
\end{enumerate} 

By rescaling, we will later reduce the analysis to the case $r=1$, so we introduce the short-hand $\tilde{Z} = \tilde{Z}_1$.
\end{definition}
 Notice that the position of the second scatterer is determined by $(u,\xi,v)\in \DD$ and $r>0$, and we do not exclude the possibility of overlap between the scatterers. In the next lemma, we summarize the geometric arguments used in~\cite{LT20} to estimate the probability of specific recollision events.

 \begin{lemma} \label{lem:LTgeom} Let $r=1$ and $(u,\xi,v) \in \DD$. Define $n\in \St$ to be a unit vector perpendicular to the plane generated by the points $0,a,b$, more precisely
 \begin{align}\label{def:n}
 	n:= \frac{a \times b}{|a| |b| \sin(\angle(a,b))},
 \end{align}
and $n_k := w_k \cdot n$ the velocity component along this direction. Then the following holds:
\begin{enumerate}
	\item $n_k$ is monotone increasing if $n_1\geq 0$, i.e. $n_k\geq n_{k-1}$ for $1\leq k\leq N$.
	\item We have the following bounds for the trapping time and angles
	\begin{align}
		\tilde{\beta} &\leq \xi  + |w_1|^{-1},  \label{ineq:beta}\\
		\angle(-e,w_j)&> \tfrac{\pi}2 - \angle (n,w_j), \label{ineq:angle} \\
		|n_2| &= |v \cdot n|\geq \tfrac12 |e \cdot (u \times v)| \label{ineq:vertical}.
	\end{align}
	\item \label{it: uTimesv} If $u,v$ are i.i.d. $\Uni(\St)$ distributed, then $w:=\frac{u\times v}{|u\times v|}$ and $\vartheta:=|u\times v|$ are independent and $w$ is again $\Uni(\St)$ distributed, while $\vartheta\sim \cf_{[0,1]} \frac{t}{\sqrt{1-t^2}} dt$.
\end{enumerate}
 \end{lemma}

 \subsection{Improved geometric estimates}\label{ssec:impro}

 The approach in this contribution is based on making quantitative improvements on the estimates in~\cite{LT20}. In Lemma~\ref{lem:LTimproved}, we give improved estimates for the probability of certain recollision events depending on the total number of collisions $N$ and the free flight time $\xi$ after the first collision. Rescaling allows us to restrict, without loss of generality, to the case $r=1$.   
 
 For recollisions with short free flight time, meaning $\xi>0$  small, the two scatterers can be close, touch, or even overlap. For such events, we will use the bounds introduced in~\cite{LT20} and reviewed above. On the other hand, for $\xi \rightarrow \infty$ we introduce the new estimates that allow us to eliminate the logarithmic divergence in the resulting estimate. Hence, we now restrict our attention to large flight times $\xi\geq 10$, where the value $10$ is chosen for convenience.

 We also remark that the estimate by Lutsko and T\'oth is sharp when $N=3$ is the total number of collisions introduced in Definition~\ref{def:recollisions}. Therefore, we will primarily consider the case $\xi\geq 10$ and $N  \geq 4$. Here the convexity of the spheres leads to improved estimates on the components $n_i$ allowing for sharper estimates on the trapping time. The estimates are closely related to the dispersive properties of Sinai-billiards introduced in~\cite{Sinai70}. The following Lemma encodes the main gain in the scattering estimates due to the convexity of the obstacles.

 \begin{lemma} \label{lem:LTimproved}
 	Let $(u,\xi,v)\in \DD$ be such that the total number of collisions satisfies $N\geq 3$, and further assume $\xi\geq 10$, and $n_1\geq 0$. Recall the definition of $n$ in~\eqref{def:n} and let $h_i := \tilde{Z}(\tau_i) \cdot n$, where $0=\tau_1 \leq \tau_1\leq \cdots \leq \tau_N$ are the collision times. Then the following inequalities hold
 	\begin{align} \label{geometric}
 		\begin{rcases}
 		h_{k+1} &\geq h_k + \frac12 \xi  n_k, \\
 		n_{k+1} &\geq n_k + \frac12 h_{k+1},
 		\end{rcases} \quad  \text{for } 1\leq k\leq N-1.
 	\end{align}
 \end{lemma}
\begin{proof}
	For the first inequality, we first observe $n_1\geq 0$ and $n_k$ monotone increasing implies $h_k\geq0$ is increasing as well. Since $\xi\geq 10$, the free flight time is bounded below by
	\begin{align*}
		\tau_{k+1}-\tau_k \geq \frac12 \xi,
	\end{align*}
	which implies the first inequality (notice that this lower bound does not hold for $\xi$ approaching zero). For the second inequality we use that the flights
	$F_{k+1} = (\tau_{k+1}-\tau_k)w_{k+1}$ and $F_k = (\tau_{k}-\tau_{k-1})w_k $ need to satisfy
	\begin{align*}
		|F_{k+1} + F_k|\leq 2,
	\end{align*}
	to lead back to the same scatterer after two consecutive flights. Using the notation $T_k=\tau_k-\tau_{k-1}$ and  $T_{k+1}=\tau_{k+1}-\tau_{k}$, we write
    \begin{align*}
        |F_{k+1}+ F_k| = |T_{k+1}\wedge T_k (w_{k+1}+w_k) + (T_k \vee T_{k+1} - T_k\wedge T_{k+1}) w_*|,
    \end{align*}
    where $w_*$ is the velocity corresponding to the larger of $T_k$, $T_{k+1}$. 
    Since $|w_k|=|w_{k+1}|$, we have 
    \begin{align*}
        T_{k+1}\wedge T_k (w_{k+1}+w_k) \cdot w_* \geq 0,
    \end{align*}
    so we obtain
	\begin{align} \label{eq:flightanglebd}
		8 |w_{k+1}+w_k| \leq T_{k+1}\wedge T_k |w_{k+1}-w_k| \leq  |F_{k+1} + F_k|\leq 2.	
	\end{align}
	Let $V_k:=\tilde{Z}(\tau_k)-c$, where $c\in \Reals^3$ is the centre of the scatterer hit in the $k$-th collision. Then, by the elastic collision rule
	\begin{align} \label{eq:collision}
		w_{k+1} = w_k -2 (V_k \cdot w_k) V_k.
	\end{align}
	Together with~\eqref{eq:flightanglebd} this yields
	\begin{align*}
		|w_k -(V_k \cdot w_k) V_k| \leq \frac18,
	\end{align*}
	which implies $V_k \cdot w_k \leq -\frac78$. Here we use that $v_k\cdot V_k\leq 0$ is necessarily negative, since a collision occurs. Inserting this inequality back into~\eqref{eq:collision} gives
	\begin{align*}
		n_{k+1}\geq n_k +\frac12 h_k,
	\end{align*}
	and finishes the proof. 
\end{proof}

 \subsection{Probability of scattering events}
\label{ssec:prob}

 On the set $\DD$ introduced in Definition~\ref{def:recollisions}, we consider the measure
\begin{align}
	\mu = \Uni(\St) \times \operatorname{EXP}(1|1) \times \Uni(\St), \label{eq:defmu}
\end{align}
where $\operatorname{EXP}(1|1)$ is the unit exponential distribution on $\Reals^+$ conditioned to  the interval $[0,1]$, i.e. the image measure on $[0,1]$ has the density
\begin{align*}
	\rho_{1,1}(x)=  \cf_{[0,1]} (x)\frac{e^{1-x}}{e-1} .
\end{align*}

A key point in~\cite{LT20} is played by the set of \emph{shadowing events} $\hat{A}_r$, where the second scatterer is placed in the path of the tracer particle before hitting the first obstacle, and~\emph{recollision events} $\tilde{A}$, where the tracer particle scatters back to the first scatterer after hitting the second. They can be formally defined as

\begin{align*}
	\hat{A}_r &:=\Bigl \{(u,\xi,v)\in \DD: \inf\{t:|\xi u +r \frac{u-v}{|u-v|}+te| <r \} <\infty \Bigr \},\\
	\tilde{A}_r &:=\Bigl  \{(u,\xi,v)\in \DD: \inf\{t:|\xi u +r \frac{u-e}{|u-e|}+tv| <r\} <\infty \Bigr \}.
\end{align*}
We will also use the related set $\tilde{A}_r'$ given by
\begin{align} \label{set:prime}
	\tilde{A}_r \subset \tilde{A}_r' &:= \Bigl \{(u,\xi,v)\in \DD: \angle(u,v)\leq 2 r \xi^{-1}\Bigr \}.
\end{align}
On the two sets respectively, we define $\hat{w}_r, \hat{\beta}_r$ and $\tilde{w}, \tilde{\beta}_r$ as
\begin{enumerate}
	\item $\hat{w}:=u$, $\hat{\beta}_r = 0$ (no trapping, second scatterer is ignored because of shadowing),
	\item $\tilde{w}$, $\tilde{\beta}_r$ are given by the trajectory of the tracer particle bouncing between the two obstacles, where $\tilde{\beta}_r$ is the time of the last collision and $\tilde{w}_r$ the outgoing velocity.
\end{enumerate}

We will improve the bounds on the likelihood of such events by estimating events with different numbers of collisions separately. Recall that for $r>0$ and $(u,\xi,v)\in \tilde{A}_r$, we denote by $N=N(r,u,\xi,v)$ the number of collisions of the tracer particle between the two obstacles.

In the rest of this section, we will derive the following geometric bounds which improve the estimates in Proposition~3 of \cite{LT20}.

	\begin{lemma} \label{lem:thmrescaled} Let $r=1$. Then for all $s> 0$ we have
	\begin{align}
		\lambda\bigl( (u,h,v)\in \tilde{A}: \angle(-e,\tilde w)< s^{-1}, \, h\leq 10 \bigr ) &\leq \frac{C}{1+s}, \label{resc:shortNany}\\
		\lambda\bigl( (u,h,v)\in \tilde{A}: \angle(-e,\tilde w)< s^{-1},\,  h\geq 10  , \, N=3 \bigr) &\leq \frac{C}{1+s}, \label{resc:longN3} \\
		\lambda\bigl( (u,h,v)\in \tilde{A}: \angle(-e,\tilde w)< s^{-1},\,  h\geq 10 , \, N\geq 4\bigr) &\leq \frac{C}{(1+s)^2}, \label{resc:longN4}\\
		\lambda\bigl( (u,h,v)\in \tilde{A}: \tilde \beta  >s, \, N =3 \bigr ) &\leq \frac{C}{1+s}, \label{resc:trapping3} \\
		\lambda\bigl( (u,h,v)\in \tilde{A}: \tilde \beta  >s , \, N\geq 4 \bigr) &\leq \frac{C}{1+s^2}, \label{resc:trapping4plus}
	\end{align}
	where the measure $\lambda$ is given by
	\begin{align*}
		\lambda = \Uni(\St) \times \operatorname{Leb}(\Reals^+) \times \Uni (\St)
	\end{align*}
	the uniform measure on $\DD$.
	\end{lemma}
As in~\cite{LT20}, the Lemma above can be used to obtain estimates for the probabilities of collision events with respect to the measure $\mu$ as defined \eqref{eq:defmu}, by the same scaling argument outlined there.
\begin{corollary}
	The following improved geometric estimates hold
	\begin{align}
		\mu\bigl( (u,h,v)\in \tilde{A}_r: \angle(-e,\tilde w)< s^{-1}, \, h\leq 10 r  \bigr) &\leq \frac{Cr}{1+s}, \label{eq:shortNany}\\
		\mu\bigl( (u,h,v)\in \tilde{A}_r: \angle(-e,\tilde w)< s^{-1},\,  h\geq 10 r , \, N=3 \bigr) &\leq \frac{Cr}{1+s}, \label{eq:longN3} \\
		\mu\bigl( (u,h,v)\in \tilde{A}_r: \angle(-e,\tilde w)< s^{-1},\,  h\geq 10 r , \, N \geq 4 \bigr ) &\leq \frac{Cr}{(1+s)^2},\label{eq:longN4}\\
		\mu\bigl( (u,h,v)\in \tilde{A}_r: \tilde \beta  >s, \, N =3 \bigr ) &\leq \frac{Cr}{1+s}, \label{eq:trapN3}\\
		\mu\bigl( (u,h,v)\in \tilde{A}_r: \tilde \beta  >s, \, N \geq 4  \bigr) &\leq \frac{Cr}{1+s^2}. \label{eq:trapN4}
	\end{align}
\end{corollary}
\begin{remark}
	The corresponding estimates in Proposition~3 of~\cite{LT20} feature an additional $|\log(s)|$ that the new estimates can avoid.
\end{remark}

\begin{proofof} [Proof of Lemma~\ref{lem:thmrescaled}]
	The proof of~\eqref{resc:shortNany} is the same as in~\cite{LT20} (Proof of (78) there), and the logarithm disappears since we restrict to $h\leq 10 $.
	
	\medskip

To bound~\eqref{resc:longN3}, we first observe that by~\eqref{ineq:angle}
\begin{align*}
	\angle(-e,\tilde w) \geq \frac{\pi}{2} -\angle(w,n) \geq \frac12 n_3 .
\end{align*}
The latter can be bounded below using~\eqref{geometric} and~\eqref{ineq:vertical}:
\begin{align} \label{ineq:n3}
	n_3 \geq  (v\cdot n) + \frac12 h (v\cdot n) \geq (1+\frac12 h) n_2 \geq \frac12 (1+\frac12 h) |e \cdot (u \times v) |  .
\end{align}
Inserting this above allows us to estimate
\begin{align*}
	&\lambda\bigl( (u,h,v)\in \tilde{A}: \angle(-e,\tilde w)< s^{-1},\,  h\geq 10, \, N=3 \bigr) \\
	\leq &\lambda\bigl( (u,h,v)\in \tilde{A}: \tfrac{1}{2} (1+\tfrac12 h) n_2 \leq s^{-1},\,  h\geq 10, \, N=3 \bigr)\\
	\leq &\lambda\bigl( (u,h,v)\in \tilde{A}': (1+\tfrac12 h) |e \cdot (u \times v) | \leq  4 s^{-1},\,  h\geq 10, \, N=3 \bigr),
\end{align*}
where in the last inequality we have also used the inclusion $\tilde{A}\subset \tilde{A}'$ in~\eqref{set:prime}. We are now in the position to estimate the last line, by making use of the representation $u\times v = \vartheta w $ in part~\ref{it: uTimesv} of Lemma~\ref{lem:LTgeom}. This reduces the claim to the integral bounds:
\begin{align}
	&\lambda\bigl( (u,h,v)\in \tilde{A}: \angle(-e,\tilde w)< s^{-1},\,  h\geq 10, \, N=3 \bigr) \nonumber\\
	\leq &C \int_{10}^\infty \int_{\mathbb S^2 }\int_0^{1/h} \cf_{\frac12 h|e\cdot W| \leq 4/(st)} t  \ud{t} \uS(W) \uud{h}  \nonumber\\
	\leq &C \int_{10}^\infty \int_0^{1/h}  \frac{8t}{1+8sht}  \ud{t}  \uud{h} 
    \leq C \int_{10}^\infty \frac{1}{sh}\int_0^{1/h}  \frac{8sht}{1+8sht}  \ud{t}  \uud{h}\nonumber\\
	\leq &C \int_{10}^\infty \frac{1}{sh^2}  \ud{h} \leq \frac{C}{1+s}.
\end{align}
The estimate~\eqref{resc:longN4} follows analogously, using that for $n_4$ we have the stronger inequality
\begin{align}\label{ineq:n4}
	n_4  \geq \tfrac12 h_4 \geq \tfrac14 h n_3 \geq \frac{1}{16} h^2 |e \cdot (u\times v)| .
\end{align}
The main improvement of our method lies in estimating~\eqref{resc:trapping3} and~\eqref{resc:trapping4plus}. First, we split into short- and long recollisions:
\begin{align*}
	\lambda\bigl( (u,h,v)\in \tilde{A}: \tau >s \bigr) &\leq \lambda\bigl( (u,h,v)\in \tilde{A}: \tau >s, \, h\leq 10\bigr) + \lambda\bigl( (u,h,v)\in \tilde{A}: \tau >s, \, h\geq 10\bigr).
\end{align*}
The first term can be estimated using the strategy of Lutsko and T\'oth. Due to the restriction $h\leq 10$ we do not obtain a logarithmic correction:
\begin{align*}
	\lambda\bigl( (u,h,v)\in \tilde{A}: \tau >s, \, h\leq 10\bigr) \leq \frac{C}{1+s}.
\end{align*} The other term we split further into
\begin{align*}
	&\lambda\bigl( (u,h,v)\in \tilde{A}: \tau >s, \, h\geq 10\bigr) \\
	= &\lambda\bigl( (u,h,v)\in \tilde{A}: \tau >s, \, h\geq 10, N=3\bigr)
	+\sum_{k=4}^\infty \lambda\bigl( (u,h,v)\in \tilde{A}: \tau >s, \, h\geq 10, N =k\bigr).
\end{align*}
We now turn to the proof of~\eqref{resc:trapping3}. Under the constraints $N =3$ and $h\geq 10$, we know that $\tau\leq 3h$, and therefore
\begin{align*}
	\lambda\bigl( (u,h,v)\in \tilde{A}: \tau >s, \, h\geq 10, N=3\bigr) \leq \int_{s/3}^\infty P( \angle(u,v) \leq (1+h)^{-1}) \leq \frac{C}{1+s}.
\end{align*}
Finally, the estimate~\eqref{resc:trapping4plus} for the contribution of higher number of collisions can be bounded using the iterated form of~\eqref{geometric}. We first observe that
\begin{align}
	n_\ell  \geq (\tfrac{h}{4})^{\ell-2} |e \cdot (u\times v)|.
\end{align}
We have $\tau \leq 2 N  h$ and $N$ collisions require
\begin{align*}
	n_{N-1} h \leq 1.
\end{align*}
Combining these estimates we obtain
\begin{align*}
	&\lambda\bigl( (u,h,v)\in \tilde{A}: \tau >s, \, h\geq 10, N =k\bigr)\\
	\leq &\int_{s/(2k)\wedge 10}^\infty  \int_{\mathbb S^2 }\int_0^{1/h} \cf\big((\frac{h}4)^{k-3} h|e\cdot W| \leq 1/t\big) t  \ud{t} \uS(W) \uud{h}\\
	\leq &C\int_{s/(2k)\wedge 10}^\infty  \int_0^{1/h}  \frac{ 4^k}{ h^{k-2} }   \ud{t} \uud{h}
	\leq C\int_{s/(2k)\wedge 10}^\infty   \frac{ 4^k}{ h^{k-1} }  \uud{h} \leq \frac{C^k}{s^{k-2}},
\end{align*}
where we use $k\geq 4$. Summing over $k\geq 4$ yields, for $s\geq 2C$ large enough
\begin{align*}
	\lambda\bigl( (u,h,v)\in \tilde{A}: \tau >s, \, h\geq 10, N \geq 4\bigr) \leq \frac{2C^2}{s^2}.
\end{align*}
Since we know the value is bounded for $0\leq s\leq C$, this finishes the proof.
\end{proofof}

\subsection{Exit distribution of simple direct recollisions} \label{ssec:indep}

The estimates in Lemma~\ref{lem:thmrescaled} identify the probability of recollision events with $N=3$, i.e. only one back-scattering, as having non-integrable decay in $s\geq 1$.
Hence, these events are one of the obstacles for the future goal of reaching timescales with $\lim_{r\rightarrow 0} T(t) r^2=0$, which will requires to avoid any logarithmic divergences. Since we believe the estimates to be sharp, we would like to examine the geometry of these collision events closer and give an idea how this problem could be overcome in the future. More precisely, we will show that the direct recollision events with very long flight times that are responsible for the logarithmic divergence essentially recouple to the free flight process, because we can prove that the exit velocity of the tagged particle becomes uniformly distributed and independent of the entry velocity as $\xi \rightarrow \infty$. 

With the notation of the last section, we recall that the exit distribution for $N=3$, is given by
\begin{align} \label{eq:colltransfered}
	\tilde{w} = v - 2 (v\cdot \omega) \omega,
\end{align}
where $\omega=\frac{Z(\tau_2)-a}{|Z(\tau_2)-a|}\in \mathbb{S}^2$ is the direction of the transfer of momentum.

\begin{lemma} \label{lem:independence}
	We have the following convergence for the conditional exit distribution in total variation norm:
	\begin{align}
		\sup_{A \subset \mathbb{S}^2} \Bigl| P(\tilde \omega \in A |\xi=R, \, u=\nu,  \, N=3) - \frac{|A|}{|\mathbb{S}^2|}\Bigr| \leq \frac{C}{1+R}.
	\end{align}
\end{lemma}
\begin{proof}
	Let $\xi=R$ and $u\in \mathbb{S}^2$ be arbitrary, and $N=3$. Then the exit velocity $\tilde{w}$ is given by~\eqref{eq:colltransfered}, where $\omega$ is a  function of $u,v$ and $\xi$. We compute the conditional probability density as
	\begin{align*}
		\rho_{\xi=R,u=\nu} (\tilde{w}) =  \frac{1}{\mathcal{Z}}(1+2d \cos \alpha )^{-1} (1+2d/\cos(\beta))^{-1} \cf_{\cos\alpha >0},
	\end{align*}
	where $d= |Z(\tau_2)- Z(\tau_1)|$ is the distance between the last two collisions, $\alpha$ is the angle between $-v$ and $\omega$, and $\beta= \angle(\omega,Z(\tau_1)-a)$. Moreover, $\mathcal{Z}$ is the surface volume
	\begin{align*}
		\mathcal{Z} = \int_{\mathbb{S}^2} \cf_{N(v,u,\xi)=3} \ud{v},
	\end{align*}
    where $N(v,u,\xi)$ is the number of collisions in the event as a function of the parameters. From this we obtain
	\begin{align*}
		\Bigl|\rho_{\xi=R,u=\nu} (\tilde{w}) - \frac{1}{\mathcal{Z}} \frac{1}{4R ^2}\Bigr| \leq \frac{C}{1+R}, \quad  \text{ for  $\tilde w \in \supp \rho_{\xi=R,u=\nu}$}.
	\end{align*}
	Since $\mathcal{Z}$ satisfies
	\begin{align*}
		\mathcal{Z} = \frac{\pi}{R^2} + O(\frac{1}{R^3}),
	\end{align*}
	we obtain
	\begin{align*}
		\Bigl|\rho_{\xi=R,u=\nu} (\tilde{w}) - \frac{1}{4\pi}\Bigr| \leq \frac{C}{1+R}, \quad \text{ for  $\tilde w \in \supp \rho_{\xi=R,u=\nu}$}.
	\end{align*}
	The convergence of the conditional density to the uniform density shows that $$\Uni(\St) (\{\tilde  w  \notin \supp \rho_{\xi=R,u=\nu}\})\leq \frac{C}{1+R}$$ as well, and the claim follows. 
\end{proof}
As discussed in the introduction, these events are not the only ones yielding error terms that become significant on timescales $T(t)\sim r^{-2} |\log r|^{-1}$. In particular some of these terms are generated by two subsequent scattering events, which could only be treated by considering a different auxiliary process with longer memory. For the moment, we only state these additional estimates on direct recollisions for application in future research.

\section{Application to the Invariance Principle of Lutsko and Toth} \label{sec:LT}

We show how the improved geometric estimates give the following proved validity of the invariance principle in \cite{LT20}. We use the notation $X^\eps$ for the empirical Random Lorentz gas, $Y$ for
the Markovian free flight process and $W$ for the standard Wiener process of variance $1$ in $\Reals^3$.

\begin{theorem} \label{thm:LT}
  Let $T=T(\eps)$  be such that $\lim_{\eps \to 0} T(\eps)=\infty $ and $\lim_{\eps \to 0} \eps^2 |\log \eps|T(\eps)=0$. Then, for any $\delta>0$,
  \begin{align}\label{eqn:thm2.1}
    \lim_{\eps \to 0} \mathbf{P}\Bigl(\sup_{0 \leq t \leq T}|X^\eps(t) - Y(t)|v> \delta \sqrt{T} \Bigr)& =0
  \end{align}
   and hence
   \begin{align}\label{eqn:thm2.2}
   \Bigl\{ t \mapsto T^{-1/2} X^\eps (Tt)  \Bigr\}& \Rightarrow  \Bigl\{ t \mapsto W(t) \Bigr\}
  \end{align}
   as $\eps \to 0$ in the averaged-quenched sense.
\end{theorem}

The proof of the corresponding result in \cite{LT20} is the core part of their paper. A key ingredient is the construction of  an auxiliary $Z$-process via coupling with the same data which ignores interaction beyond direct shadowing and direct recollisions with the last seen scatterer. 
 The  difference between the Lorentz exploration process and the $Z$ process
is estimated in \cite[Prop. 1]{LT20} using various geometric estimates. Our improved geometric estimates above yield improved estimates for the auxiliary process: 
\begin{proposition} \label{prop1LT} There exists $C < \infty$ such that for arbitrary incoming and outgoing velocities
\begin{align}\label{eqn:49}
  \mathbf{P}\Bigl( \mathcal{X}(t)\not{\equiv} Z(t) : 0^- <t< \theta^+ \Bigr) & \leq C \eps^2 |\log \eps|.
\end{align}
\end{proposition}
This improved Proposition is enough to obtain the longer time scales:

\begin{proofof}[Proof of Theorem \ref{thm:LT}] Using the same arguments as in \cite{LT20}, Proposition \ref{prop1LT}
allows an improvement of  the stopping time estimates \cite[Lemma 9]{LT20} to a timescale $\lim_{\eps \to 0} T(\eps)=\infty $
and $\lim_{\eps \to 0} \eps^2 |\log \eps|T(\eps)=0$. The other main ingredient  \cite[Lemma 10]{LT20} holds already on longer timescales. Both Lemmas together yield \eqref{eqn:thm2.1}, which then implies \eqref{eqn:thm2.2} due to the invariance principle of the Markovian flight process.
\end{proofof}

\begin{remark} \label{rem:indirect}
    The inequality~\eqref{eqn:49} is sharp in the sense that mismatches between the two processes do occur with probability $\sim r^2 |\log(r)|$. To see this, consider two consecutive free flights $Y_i=\xi_i \omega_i$ ($i=1,2$) given by independent random variables $\xi_i \sim \Exp(1)$, $\omega_i \sim \Uni(\St)$. Then the probability of entering an $\eps$ neighborhood of the origin after the two flights satisfies:
    \begin{align*}
        P(|Y_1+Y_2|\leq \eps) \geq c \eps^2 |\log \eps|.
    \end{align*}
    This means mechanical inconsistencies appear due to \emph{indirect} recollisions with this frequency, which cannot be taken into account by the auxiliary process $Z$. 
\end{remark}

It remains to show Proposition \ref{prop1LT}.   The key challenge is to improve \cite[(66)]{LT20}, which deals with the case that there is exactly one $k$ with direct shadowing  ($\hat{\eta}_j =1$)  or a direct recollision with the directly previous scatterer ($\tilde{\eta}_k=1$) but not both.  All other cases have already been estimated to be of higher order in \cite{LT20}:

\begin{lemma} \label{lem:66impliesprop}
  Suppose that the Lorentz exploration process and the auxiliary process are constructed as in \cite{LT20} with $\gamma \in \Naturals$ different intervals. Furthermore, assume
  \begin{align}\label{eqn:66}
     \mathbf{P}\Bigl( \{\mathcal{X}(t)\not{\equiv} Z(t) : 0^- <t< \theta^+\} \cap \{  \sum_j^\gamma \eta_j =1, \} \Bigr) & \leq C \eps^2 |\log \eps|,
  \end{align}
    then Proposition \ref{prop1LT} holds.
\end{lemma}
\begin{proof}
  The proof of the corresponding proposition in \cite{LT20} is based on bounds \cite[(64),(65),(66)]{LT20} for the cases that $\sum_{j=1}^\gamma  \eta_j$  is  $>1$, $=0$ and $=1$ respectively. The first two cases have already bounds consistent with Proposition \ref{prop1LT}. The new bound in \eqref{eqn:66} then yields the improved  Proposition \ref{prop1LT}.
\end{proof}
Next, we outline the necessary changes to obtain \eqref{eqn:66} and reduce the order to $\eps^2 |\log \eps|$. As a first step, we improve the estimates (82) and (83) in \cite[Corollary 2]{LT20}.
\begin{lemma} There exist a constant $C < \infty$ such that
  \begin{align} \nonumber
  &\mathbf{P}\Bigl(\{ \hat{\eta}_k=0\}  \cap \{ \tilde{\eta}_k=1\} \cap
  \{\min_{\tau_{k-2} \leq t \leq \tau_k} \left| Z^{(k)}(t)- Z^{(k)}(\tau_{k-3}) \right| <s \}\Bigr)\\
  &  \quad  \leq C \eps \max( s,\eps |\log \eps|), \label{eqn:82} \\ \nonumber
 & \mathbf{P}\Bigl( \{ \hat{\eta}_k=0\}  \cap \{ \tilde{\eta}_k=1\} \cap
  \{\min_{\tau_{k-3} \leq t \leq \tau_{k-1} + \tilde{\beta}} \left| Z^{(k)}(t)- Z^{(k)}(\tau_{k-3}) \right| <s \} \Bigr) \\& \quad \leq C \eps \max( s | \log s| ,\eps |\log \eps|). \label{eqn:83}
\end{align}
\end{lemma}
\begin{proof}
  We outline the steps that need to be changed. For \eqref{eqn:82}, we need to consider \cite[(104)]{LT20}, where the main term to estimate is:
\begin{align}\nonumber&
\int_{\tilde{A}_\eps}  \min(\frac{s}{\angle(-e,\tilde w)},1) d \mu(u,h,v)\\& \leq
\int_{\tilde{A}_\eps,\angle(-e,\tilde w) \leq s } d \mu(u,h,v) + \int_s^\pi \int_{\tilde{A}_\eps, \angle(-e,\tilde w) =x} \frac{s}{x} d \mu(u,h,v) d  x \nonumber
 \\
  &\leq  C  \eps s +   \int_s^\pi \int_{\tilde{A}_\eps, \angle(-e,\tilde w) \leq x} \frac{s}{x} d \mu(u,h,v) d  x \leq C  \eps s +   C s \eps \int_s^\pi  \frac{1}{1+\frac{1}{x}} \frac{1}{x} d x\nonumber
 \\
  & = C  \eps s + C s \eps \left[ \log(1+\pi)-  \log(1+s) \right] \leq C \eps s,\label{eqn:104}
\end{align}
where we used \eqref{eq:shortNany}, \eqref{eq:longN3} and \eqref{eq:longN4} to estimate the angles. All other estimates to prove  \cite[Corollary 2, (82)]{LT20}  are unchanged.

For \eqref{eqn:83}, we need to adapt \cite[(107)]{LT20}, where the term yielding squared logarithmic corrections involves the trapping time $\beta$ using \eqref{eq:trapN3} and \eqref{eq:trapN4}. After scaling with $\eps$ we are integrating over the collision time parameterised by $x$, the volume of the outgoing velocities can be bounded by  $ \frac{1}{1+ (x/\eps)^2} $, then we obtain 
\begin{align}\nonumber
   &\mathbf{P}\Bigl( \{ \hat{\eta}_k=0\}  \cap \{ \tilde{\eta}_k=1\} \cap
  \{\xi_k<4 \beta\} \Bigr) \\ \nonumber
   \leq  &C \int_{0}^{\infty} \frac{\eps}{1+x} \exp(-x) \frac{1}{1+ (x/\eps)^2}  dx \\ \nonumber
   \leq &  C \int_{0}^{\infty} \frac{\eps}{1+ (x/\eps)^2}    dx\\ 
   \leq &C \eps^2. \label{eqn:107}
\end{align}
We do not change any other estimates, such that we obtain \eqref{eqn:83}.
\end{proof}
It is not easily possible to improve (102) in \cite{LT20}, such that we cannot remove the other $|\log \eps|$ term in \eqref{eqn:82} despite obtaining better estimates here.
We are now in the position to sketch how to provide the final missing estimate.
\begin{lemma}
  There exist a constant $C < \infty$ such that such that \eqref{eqn:66} holds.
\end{lemma}
\begin{proof}
  The improvements to \cite[(82),(83)]{LT20} by \eqref{eqn:82} and \eqref{eqn:83} then directly improve \cite[(109)]{LT20} to $C \gamma \eps^2 |\log \eps|$.
The final missing ingredient  are the estimates in \cite[(110)]{LT20}, these involve Green's functions estimates as well as estimates involving \eqref{eqn:82} and \eqref{eqn:83}, which directly yield improved estimates in \cite[(113)]{LT20}, which yield estimates of order $\eps^2 |\log \eps|$ for \cite[(110)]{LT20} and hence for \cite[(66)]{LT20}, i.e. \eqref{eqn:66} holds.
\end{proof}

This ensures the condition for Lemma \ref{lem:66impliesprop}, which yields Proposition \ref{prop1LT}, which in turn yielded the invariance principle in Theorem~\ref{thm:LT}.

\subsection*{Acknowledgements}
This work was partially supported through  The Leverhulme Trust research project grant  RPG-2020-107.  The authors would like to thank the Isaac Newton Institute for Mathematical Sciences, Cambridge, for support and hospitality during the programme \emph{Frontiers in kinetic theory}  where this work was initiated.

\bibliographystyle{plain}
\bibliography{LorentzKR.bib}

\begin{thebibliography}{10}

\bibitem{Balint23}
P\'eter B\'alint, Henk Bruin, and Dalia Terhesiu.
\newblock Periodic {L}orentz gas with small scatterers.
\newblock {\em Probab. Theory Related Fields}, 186(1-2):159--219, 2023.

\bibitem{bod15brown}
Thierry Bodineau, Isabelle Gallagher, and Laure Saint-Raymond.
\newblock The {B}rownian motion as the limit of a deterministic system of
  hard-spheres.
\newblock {\em Invent. Math.}, 203(2):493--553, 2015.

\bibitem{Bodineau23}
Thierry Bodineau, Isabelle Gallagher, Laure Saint-Raymond, and Sergio
  Simonella.
\newblock Long-time correlations for a hard-sphere gas at equilibrium.
\newblock {\em Comm. Pure Appl. Math.}, 76(12):3852--3911, 2023.

\bibitem{Bodineau23-2}
Thierry Bodineau, Isabelle Gallagher, Laure Saint-Raymond, and Sergio
  Simonella.
\newblock Dynamics of dilute gases at equilibrium: from the atomistic
  description to fluctuating hydrodynamics.
\newblock {\em Ann. Henri Poincar\'e}, 25(1):213--234, 2024.

\bibitem{Bodineau24}
Thierry Bodineau, Isabelle Gallagher, Laure Saint-Raymond, and Sergio
  Simonella.
\newblock Long-time derivation at equilibrium of the fluctuating {B}oltzmann
  equation.
\newblock {\em Ann. Probab.}, 52(1):217--295, 2024.

\bibitem{BBS83}
C.~Boldrighini, L.~A. Bunimovich, and Ya.~G. Sinai.
\newblock On the {B}oltzmann equation for the {L}orentz gas.
\newblock {\em J. Statist. Phys.}, 32(3):477--501, 1983.

\bibitem{Cag10}
E.~Caglioti and F.~Golse.
\newblock On the {B}oltzmann-{G}rad limit for the two dimensional periodic
  {L}orentz gas.
\newblock {\em J. Stat. Phys.}, 141(2):264--317, 2010.

\bibitem{Durr-87}
D.~D\"urr, S.~Goldstein, and J.~L. Lebowitz.
\newblock Asymptotic motion of a classical particle in a random potential in
  two dimensions: {L}andau model.
\newblock {\em Comm. Math. Phys.}, 113(2):209--230, 1987.

\bibitem{Fou24}
Florent Foug\`eres.
\newblock On the derivation of the linear {B}oltzmann equation from the
  nonideal {R}ayleigh gas.
\newblock {\em J. Stat. Phys.}, 191(10):Paper No. 136, 16, 2024.

\bibitem{Gal70}
G.~Galavotti.
\newblock Rigorous theory of {B}oltzmann equation in the {L}orentz gas.
\newblock {\em Nota interna}, 358, 1970.

\bibitem{golse22}
F.~Golse.
\newblock The {B}oltzmann-{G}rad limit for the {L}orentz gas with a {P}oisson
  distribution of obstacles.
\newblock {\em Kinet. Relat. Models}, 15(3):517--534, 2022.

\bibitem{Kesten81}
H.~Kesten and G.~C. Papanicolaou.
\newblock A limit theorem for stochastic acceleration.
\newblock {\em Comm. Math. Phys.}, 78(1):19--63, 1980/81.

\bibitem{LT20}
Christopher Lutsko and B\'alint T\'oth.
\newblock Invariance principle for the random {L}orentz gas---beyond the
  {B}oltzmann-{G}rad limit.
\newblock {\em Comm. Math. Phys.}, 379(2):589--632, 2020.

\bibitem{Toth21}
Christopher Lutsko and B\'alint T\'oth.
\newblock Invariance principle for the random wind-tree process.
\newblock {\em Ann. Henri Poincar\'e}, 22(10):3357--3389, 2021.

\bibitem{lutsko24}
Christopher Lutsko and Balint Toth.
\newblock Diffusion of the random {L}orentz process in a magnetic field, 2024.
\newblock arXiv:2411.03984.

\bibitem{marklof08}
Jens Marklof and Andreas Str\"ombergsson.
\newblock Kinetic transport in the two-dimensional periodic {L}orentz gas.
\newblock {\em Nonlinearity}, 21(7):1413--1422, 2008.

\bibitem{Marklof11}
Jens Marklof and Andreas Str\"ombergsson.
\newblock The {B}oltzmann-{G}rad limit of the periodic {L}orentz gas.
\newblock {\em Ann. of Math. (2)}, 174(1):225--298, 2011.

\bibitem{Marklof-24}
Jens Marklof and Andreas Str\"ombergsson.
\newblock Kinetic theory for the low-density {L}orentz gas.
\newblock {\em Mem. Amer. Math. Soc.}, 294(1464):v+136, 2024.

\bibitem{Marklof16}
Jens Marklof and B\'alint T\'oth.
\newblock Superdiffusion in the periodic {L}orentz gas.
\newblock {\em Comm. Math. Phys.}, 347(3):933--981, 2016.

\bibitem{Matthies18}
Karsten Matthies, George Stone, and Florian Theil.
\newblock The derivation of the linear {B}oltzmann equation from a {R}ayleigh
  gas particle model.
\newblock {\em Kinet. Relat. Models}, 11(1):137--177, 2018.

\bibitem{MatthiesSyntaka2024surrey}
Karsten Matthies and Theodora Syntaka.
\newblock Derivation of kinetic and diffusion equations from a hard-sphere
  rayleigh gas using collision trees and semigroups, 2024.
\newblock arXiv:2405.04449v1, to appear in LMS lecture notes.

\bibitem{matthies24}
Karsten Matthies and Theodora Syntaka.
\newblock Fractional diffusion as the limit of a short range potential
  {R}ayleigh gas, 2024.
\newblock arXiv:2405.19025.

\bibitem{nota24}
Alessia Nota, Dominik Nowak, and Chiara Saffirio.
\newblock Diffusion limit of the low-density magnetic lorentz gas, 2024.
\newblock arXiv:2412.11134.

\bibitem{NSW21}
Alessia Nota, Juan J.~L. Vel\'azquez, and Raphael Winter.
\newblock Interacting particle systems with long-range interactions: scaling
  limits and kinetic equations.
\newblock {\em Atti Accad. Naz. Lincei Rend. Lincei Mat. Appl.},
  32(2):335--377, 2021.

\bibitem{NSW22}
Alessia Nota, Juan J.~L. Vel\'azquez, and Raphael Winter.
\newblock Interacting particle systems with long-range interactions:
  approximation by tagged particles in random fields.
\newblock {\em Atti Accad. Naz. Lincei Rend. Lincei Mat. Appl.},
  33(2):439--506, 2022.

\bibitem{Nota19}
Alessia Nota, Raphael Winter, and Bertrand Lods.
\newblock Kinetic description of a {R}ayleigh gas with annihilation.
\newblock {\em J. Stat. Phys.}, 176(6):1434--1462, 2019.

\bibitem{Sinai70}
Ja.\~G. Sina\u~i.
\newblock Dynamical systems with elastic reflections. {E}rgodic properties of
  dispersing billiards.
\newblock {\em Uspehi Mat. Nauk}, 25(2(152)):141--192, 1970.

\bibitem{Spohn78}
Herbert Spohn.
\newblock The {L}orentz process converges to a random flight process.
\newblock {\em Comm. Math. Phys.}, 60(3):277--290, 1978.

\bibitem{spohn91}
Herbert Spohn.
\newblock {\em Large Scale Dynamics of Interacting Particles}.
\newblock Texts and Monographs in Physics. Springer Berlin Heidelberg, 1991.

\bibitem{toth25}
Bálint Tóth.
\newblock Semi-quenched invariance principle for the random lorentz gas: Beyond
  the boltzmann–grad limit.
\newblock {\em Entropy}, 27(4), 2025.

\bibitem{wennberg23}
Bernt Wennberg.
\newblock The {L}orentz gas with a nearly periodic distribution of scatterers.
\newblock {\em J. Stat. Phys.}, 190(7):Paper No. 123, 27, 2023.

\end{thebibliography}

\end{document}